\documentclass[journal]{IEEEtran}
\usepackage{cite}
\usepackage{amsmath,amssymb,amsfonts}
\usepackage{algorithmic}
\usepackage{graphicx}
\usepackage[percent]{overpic}
\usepackage{textcomp}
\usepackage[dvipsnames]{xcolor}
\usepackage{mathtools}
\usepackage[english]{babel}

\usepackage{amsthm}

\theoremstyle{plain}
\newtheorem{theorem}{Theorem}
\newtheorem{lemma}{Lemma}
\newtheorem{proposition}{Proposition}

\theoremstyle{definition}
\newtheorem{remark}{Remark}
\newtheorem{definition}{Definition}
\newtheorem{assumption}{Assumption}

\usepackage{stmaryrd}
\usepackage{dsfont}
\usepackage{multirow}

\def\buildextended{} 
\newif\ifextendedversion
\ifdefined\buildextended
  \extendedversiontrue
\else
  \extendedversionfalse
\fi
\newcommand{\proofpointer}[2]{%
  \ifextendedversion
    See Appendix~\ref{#1}.%
  \else
    #2 The full proof is provided in the extended arXiv version of this paper.%
  \fi
}

\newcommand{\vmnom}{\vm^{\bullet}}
\newcommand{\vvnom}{\vv^{\bullet}}

\let\ForAll\forall
\renewcommand\forall{\ForAll\,}

\usepackage{amsmath,amssymb}

\newcommand{\R}{\mathbb{R}}

\renewcommand{\Re}[1]{\operatorname{Re}\left\{#1\right\}}

\newcommand{\vct}[1]{\boldsymbol{#1}}
\newcommand{\mtx}[1]{\boldsymbol{#1}}

\newcommand{\T}{\top}

\newcommand{\Range}{\operatorname{range}}

\newcommand{\norm}[1]{\left|\left|#1\right|\right|}

\newcommand{\abs}[1]{\left|#1\right|}

\newcommand{\set}[1]{\mathcal{#1}}

\newcommand{\vd}{\vct{d}}

\newcommand{\vk}{\vct{k}}

\newcommand{\vm}{\vct{m}}

\newcommand{\vp}{\vct{p}}
\newcommand{\vq}{\vct{q}}

\newcommand{\vs}{\vct{s}}

\newcommand{\vv}{\vct{v}}

\newcommand{\vz}{\vct{z}}
\newcommand{\valpha}{\vct{\alpha}}

\newcommand{\vtheta}{\vct{\theta}}

\newcommand{\vkappa}{\vct{\kappa}}

\newcommand{\vxi}{\vct{\xi}}

\newcommand{\vsigma}{\vct{\sigma}}

\newcommand{\vomega}{\vct{\omega}}
\newcommand{\vzero}{\vct{0}}
\newcommand{\vone}{\vct{1}}

\newcommand{\mK}{\mtx{K}}

\newcommand{\mR}{\mtx{R}}
\newcommand{\mS}{\mtx{S}}

\newcommand{\mX}{\mtx{X}}

\newcommand{\vvbar}{\underline{\vct{v}}}
\newcommand{\vvubar}{\overline{\vct{v}}}
\newcommand{\vmbar}{\underline{\vct{m}}}
\newcommand{\vmubar}{\overline{\vct{m}}}

\newcommand{\setA}{\set{A}}

\newcommand{\setC}{\set{C}}

\newcommand{\setE}{\set{E}}

\newcommand{\setK}{\set{K}}

\newcommand{\setN}{\set{N}}

\newcommand{\setP}{\set{P}}

\newcommand{\sgn}{\operatorname{sign}}

\newcommand{\DVDP}{\mR}
\newcommand{\DVDQ}{\mX}

\newcommand{\diag}{\operatorname{diag}}

\newcommand{\alphamin}{\underline{\alpha}}

\newcommand{\valphamin}{\underline{\valpha}}

\newcommand{\mStilde}{\mtx{\tilde{S}}}

\newcommand{\wtilde}{\tilde{w}}

\def\QEDclosed{\mbox{\rule[0pt]{1.3ex}{1.3ex}}} 

\def\QED{\QEDclosed} 

\usepackage[hidelinks]{hyperref}

\usepackage{graphicx} 

\usepackage{booktabs} 
\usepackage{array} 
\usepackage{paralist} 
\usepackage{verbatim} 
\usepackage{steinmetz}

\usepackage{cite}

\def\BibTeX{{\rm B\kern-.05em{\sc i\kern-.025em b}\kern-.08em
    T\kern-.1667em\lower.7ex\hbox{E}\kern-.125emX}}
\usepackage{balance}
\renewcommand{\baselinestretch}{1}

\DeclareSymbolFont{sfoperators}{OT1}{cmss}{m}{n}
\DeclareSymbolFontAlphabet{\mathsf}{sfoperators}

\makeatletter
\def\operator@font{\mathgroup\symsfoperators}
\makeatother

\begin{document}

\title{Zero-Sum Power Factor Games}
\author{Cameron Khanpour, Samuel Talkington, Mathieu Dahan, and Daniel K. Molzahn
\thanks{This material is based upon work supported by the National Science Foundation Graduate Research
Fellowship Program under Grant Nos. DGE-1650044 and DGE-2039655 and NSF Award \#2145564. Any opinions, findings, and conclusions
or recommendations expressed in this material are those of the authors and do not necessarily reflect the views of the National Science Foundation.}
\thanks{C. Khanpour and D. K. Molzahn are with the School of Electrical and Computer Engineering, Georgia Institute of Technology, Atlanta, GA, USA. (e-mail: \{khanpour, molzahn\}@gatech.edu)}
\thanks{S. Talkington was with the School of Electrical and Computer Engineering, Georgia Institute of Technology, Atlanta, GA, USA. He is now with the Department of Electrical Engineering and Computer Science, University of Michigan, Ann Arbor, MI, USA. (e-mail: talks@umich.edu)} 
\thanks{M. Dahan is with the School of Industrial and Systems Engineering, Georgia Institute of Technology, Atlanta, GA, USA. (e-mail: mathieu.dahan@isye.gatech.edu)}%
}

\maketitle

\begin{abstract}
Variable active power injections arising from device behavior or compromised dispatch complicate voltage regulation in electric power networks with distributed energy resources (DERs). An operator can limit the resulting voltage deviations by remotely selecting DER reactive power parameters before observing the active power injections. IEEE Standard 1547-2018 specifies constant power factor as one such control mode, coupling each device's reactive power to its realized active power. Using a linear voltage model, we formulate the operator's decision as a robust minimax problem in which the operator minimizes the largest feasible aggregate voltage deviation. We solve this problem by expressing the power factor decisions through continuous reactive to active power ratios and exactly decomposing the payoff according to the signs of the voltage deviations. When every feasible voltage residual remains on its initial side of nominal, the resulting ratios cancel each injection's contribution and yield a closed form minimax strategy. We identify realistic DER ratings for which this strategy applies and quantify the regulation capacity lost under restricted power factor ranges. Numerical tests check the cancellation computation, solve the complete minimax problem directly at a representative DER rating, and compare the linear voltage predictions with nonlinear AC power flow.
\end{abstract}


\section{Introduction}
\label{sec:introduction}
\ifextendedversion
\label{sec:high-level-prob-description}
\IEEEPARstart{R}{eliable} operation of alternating current electricity networks requires maintaining nodal voltages close to nominal values. ANSI Standard C84.1-2020 generally restricts service voltages to within $\pm 5\%$ of nominal \cite{ansi_voltage_standard_2020}. Integrating DERs such as solar photovoltaics, energy storage, and electric vehicles complicates this task because their active power injections vary with customer dispatch, aggregator decisions, curtailment, and site conditions \cite{srivastava_voltage_2023,arnold_adaptive_2022,robbins_two-stage_2013}.

The DER managing entity, termed the \emph{operator}, regulates voltage by configuring reactive power control. IEEE Standard 1547-2018 requires these control parameters to be remotely readable and writable and specifies constant power factor as a control mode \cite{noauthor_ieee_2018}. Under this mode, each installed power factor parameter couples the device's reactive power response to its realized active power. The power factor setting must therefore remain effective across active power injections that may be unknown when the operator configures the device.

Conventional voltage regulation methods often optimize active or reactive power separately, while newer methods coordinate both components \cite{arnold_optimal_2016,robbins_optimal_2016,gupta_model-less_2022,mahmoodi_dynamic_envelopes_2023}. These approaches do not address the case in which the operator must choose a fixed reactive power setting before another entity determines active power. This timing matters under both routine uncertainty and compromised dispatch due to a cyberattack. The DER apparent power limit further creates a physical tradeoff between active power capacity and reactive power support.

We therefore ask how the operator should select fixed power factor parameters before observing active power so that voltage regulation remains effective over the feasible DER behavior. We also ask when the physical structure yields an analytical solution and how device ratings and restrictions on power factor affect the achievable regulation performance. Fig.~\ref{fig:pf-game-highlevel} illustrates the information pattern.

\begin{figure}[t]
    \centering
    \includegraphics[width=0.84\linewidth,keepaspectratio]{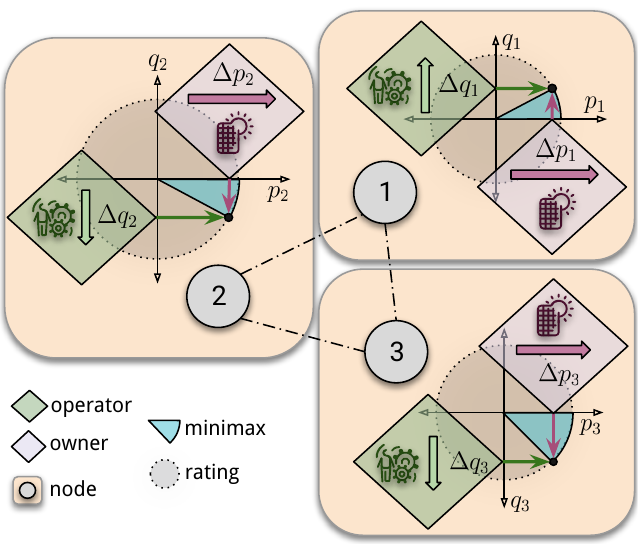}
    \caption{The operator installs the reactive power parameters before the active power player realizes the injections. Constant power factor control then couples the active and reactive powers.}
    \label{fig:pf-game-highlevel}
\end{figure}

\subsection{Contributions} 
We formulate the operator's decision as a robust minimax design. The operator first selects one constant power factor setting for each DER. The maximizing active power player, termed the \emph{owner}, then selects an injection vector within the device limits. The operator minimizes the largest aggregate absolute voltage deviation that the owner can produce. In the robust interpretation, the inner maximization identifies the worst feasible realization, and no strategic or malicious actor is required. In the cybersecurity interpretation, the same formulation is a zero-sum Stackelberg game in which a stealthy attacker observes the installed power factor settings and then selects compromised active power references. The reactive power parameter channel retains integrity, and the injections remain consistent with the commanded power factor setting to avoid immediate detection. Both interpretations lead to the same minimax problem with the operator moving first.

This paper makes the following three contributions.
\begin{enumerate}
    \item We formulate constant power factor control as a robust minimax design while enforcing each DER's apparent power limit. Because the power factor determines both the reactive response and the feasible active power set, the model captures capacity effects omitted when active and reactive power are treated independently. The formulation applies to robust voltage regulation and related hosting capacity and curtailment assessments.
    \item We derive the solution structure by representing the power factor decisions with continuous reactive-to-active power ratios (Lemma~\ref{lemma:k-transform}) and exactly decomposing the payoff by the sign pattern of the voltage residual (Proposition~\ref{prop:orthant-wise-affine-representation}). At fixed operator settings, each worst-case active power injection is either zero or its largest feasible value, depending on whether increasing that injection raises or lowers the aggregate voltage deviation (Proposition~\ref{prop:owner-best-response-fixed-kappa}). When no bus voltage crosses its nominal value, the optimal minimax ratio cancels the corresponding injection's contribution and yields a closed form strategy (Theorem~\ref{thm:closed-form-minimax}).
    \item We characterize how DER ratings and restrictions on the admissible power factor range affect regulation performance. Numerical tests check the cancellation computation, solve the complete minimax problem directly at a representative DER rating, and compare the linear voltage predictions with nonlinear AC power flow. The analytical power factors also resemble empirical recommendations in \cite{osti_1431468,rylander_methods_2016}, providing a game theoretic characterization.
\end{enumerate}

\subsection{Literature Review}
\subsubsection{Competitive reactive power control and voltage regulation}
Optimal reactive power control is a fundamental problem in electric power system engineering, both with network models \cite{arnold_adaptive_2022,arnold_optimal_2016,robbins_two-stage_2013,robbins_optimal_2016} and without them \cite{gupta_model-less_2022}. IEEE Standard 1547-2018 mandates that reactive power control parameters ``\textit{shall be available for reading and writing}" by DER managing entities through a communication interface \cite[10.6]{noauthor_ieee_2018}. This ability to remotely set the parameters is a key modeling assumption in our formulation.

Some of the early work in this area developed a two-stage control theoretic approach to controlling reactive power for voltage regulation \cite{robbins_two-stage_2013}. Augmentations of the voltage regulation problem with incentive-based considerations have emerged in the recent literature, in both offline \cite{zhou_incentive-based_2017} and online \cite{zhou_incentive-based_2018} optimization frameworks. 

Some of the work closest to ours studies noncooperative volt-var control \cite{zhou_local_2016,zhou_incentive-based_2017,zhou_incentive-based_2018,zhou_reverse_2021}. A noncooperative voltage regulation game is solved iteratively using fixed point and contraction arguments in \cite{zhou_local_2016}. Our work instead considers a robust minimax design in which the power factor is fixed before the active power realization is observed. Enforcing the apparent power limit also makes the operator anticipate how its setting changes both the voltage response and the feasible active power set, an effect related to signal anticipation in voltage regulation \cite{liu_signal-anticipation_2020}.

Recent work also considers structural interdependence between active and reactive power in game-theoretic settings~\cite{yang_prosumer-driven_2022,wang_reactive_2018}. A smoothed Nikaido-Isoda descent direction algorithm computes an equilibrium of a noncooperative Nash formulation for coordinated real and reactive power control~\cite{yang_prosumer-driven_2022}. A learning algorithm for a prospect theoretic reactive power compensation game converges to a mixed strategy Nash equilibrium \cite{wang_reactive_2018}. Both approaches compute equilibria iteratively. Our formulation asks a different question because the operator commits a power factor before observing active power and protects against the worst feasible subsequent realization. Our model also commits the constant power factor specified in IEEE Standard 1547-2018, which directly couples the active and reactive powers \cite{noauthor_ieee_2018}.

\subsubsection{Other non-cooperative games in infrastructure networks}
Game theory has been applied to reactive power control~\cite{arif_cooperative_2017,zhou_local_2016,farivar_equilibrium_2013,yang_prosumer-driven_2022,del_nozal_game-theoretic_2018}, on-load tap changer switching~\cite{tasnim_game_2023}, and the scheduling of loads, distributed generation, storage, and other flexible resources~\cite{scarabaggio_noncooperative_2022}. Strategic monitoring of power and water networks has also been formulated as a game \cite{milosevic_strategic_2024}. These works illustrate the broader role of game-theoretic models in infrastructure control and protection.

There is also extensive literature on \emph{economic} active-reactive power interactions. Bilevel optimization has been used to model separate entities controlling active and reactive power~\cite{almeida_optimal_2011}. Competitive active and reactive power scheduling frameworks have also been developed for electricity markets~\cite{jiang_strategic_2023,menati_competitive_2022}. Economic incentive mechanisms offer another approach for coordinating active and reactive power~\cite{yang_prosumer-driven_2022}.

Finally, our work is related to cybersecurity games for network systems. A tri-level defender-attacker-defender game models network security with compromised DERs represented by false complex power set points \cite{devendra_security_DER_2017}. In contrast, we study the complementary regime in which the certified reactive power control channel retains integrity, so that a
\textit{stealthy} adversary is restricted to injections consistent with the
commanded power factor setting (see Remark~\ref{rem:threat-model}). Section~\ref{sec:limitations:active-power-strategy} discusses the fully compromised regime as future work. Routing and network interdiction games provide related zero-sum models over network flows \cite{dahan_flow_routing_2015,dahan_probability_2022}. In contrast, our game is posed over complex nodal power contributions with an operator-dependent feasible set.

\subsection{Article Outline}
\label{sec:paper-outline}
Section~\ref{sec:prelim} develops the network model, voltage approximation, and robust power factor design. Section~\ref{sec:analytical-results} decomposes the lifted payoff by voltage orthant, characterizes the optimal minimax strategy in closed form when the voltage residual remains in the nominal orthant, and interprets the resulting power factors. Section~\ref{sec:numerical-results} verifies the results on realistic network models. Section~\ref{sec:discussion} concludes the paper.

\else
\label{sec:high-level-prob-description}

\IEEEPARstart{S}{ecure} operation of alternating current electricity networks requires maintaining nodal voltages close to nominal values. ANSI Standard C84.1-2020 generally restricts service voltages to within $\pm5\%$ of nominal \cite{ansi_voltage_standard_2020}. Active power injections from DERs complicate this task because they vary with customer dispatch, aggregator decisions, curtailment, and site conditions \cite{srivastava_voltage_2023,arnold_adaptive_2022,robbins_two-stage_2013}.

The DER managing entity, termed the \emph{operator}, regulates voltage by configuring reactive power control. IEEE Standard 1547-2018 requires these control parameters to be remotely readable and writable and specifies constant power factor as a control mode \cite{noauthor_ieee_2018}. Under this mode, each installed power factor parameter couples the device's reactive power response to its realized active power. The setting must therefore remain effective across active power injections that may be unknown when the operator configures the device.

Conventional voltage regulation methods often optimize active or reactive power separately, while newer methods coordinate both components \cite{arnold_optimal_2016,robbins_optimal_2016,gupta_model-less_2022}. These approaches do not address the case in which the operator must choose a fixed power factor before another entity determines active power. This timing matters under both routine uncertainty and compromised dispatch. We therefore ask how the operator should select fixed power factor parameters before observing active power so that voltage regulation remains effective over the feasible DER behavior. Fig.~\ref{fig:pf-game-highlevel} illustrates this information pattern.

\begin{figure}[t]
    \centering
    \includegraphics[width=0.84\linewidth,keepaspectratio]{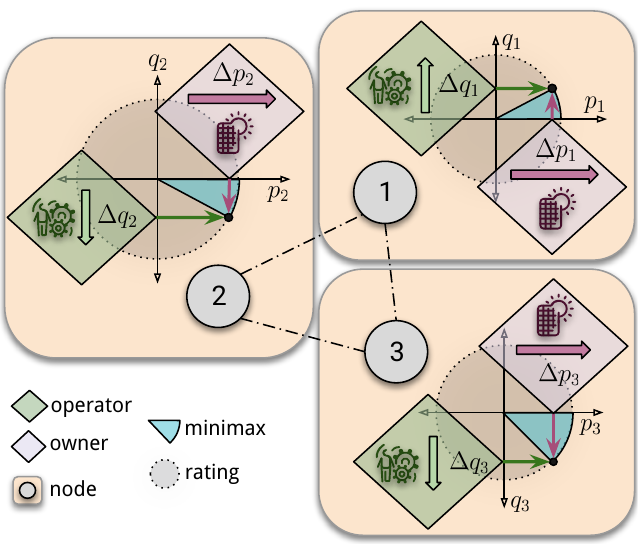}
    \caption{The operator installs the reactive power parameters before the active power player realizes the injections. Constant power factor control then couples the active and reactive powers.}
    \label{fig:pf-game-highlevel}
\end{figure}

We formulate the operator's decision as a robust minimax design. The operator first selects one constant power factor setting for each DER. The maximizing active power player, termed the \emph{owner}, then selects an injection vector within the device limits. The operator minimizes the largest aggregate absolute voltage deviation that the owner can produce. In the robust interpretation, this maximization identifies the worst feasible realization, and no strategic or malicious actor is required. In the security interpretation, the same formulation is a zero-sum Stackelberg game in which a stealthy attacker observes the installed power factor settings and then selects compromised active power references. The reactive power parameter channel retains integrity, and the injections remain consistent with the commanded power factor setting to avoid immediate detection.

Related work has developed noncooperative and incentive-based formulations for volt-var control \cite{zhou_local_2016,zhou_incentive-based_2017,zhou_incentive-based_2018} and game theoretic models of active and reactive power interactions \cite{wang_reactive_2018,yang_prosumer-driven_2022}. The equilibrium methods for the latter models are iterative. Our formulation asks a different question because the operator commits the power factors before observing active power and protects against the worst feasible subsequent realization. A centralized optimal power flow can optimize both power components under a common objective, but it does not represent this decision under uncertainty. Security games with compromised DERs provide a complementary fully compromised model \cite{devendra_security_DER_2017}. We instead study the regime in which active power dispatch is compromised or uncertain while the certified reactive power control channel remains intact.

This paper makes three contributions. First, we formulate constant power factor control as a robust minimax design while enforcing each DER's apparent power limit. The power factor determines both the reactive response and the feasible active power set. Second, we derive the solution structure using continuous ratios between reactive and active power (Lemma~\ref{lemma:k-transform}) and an exact decomposition by the signs of the voltage deviations (Proposition~\ref{prop:orthant-wise-affine-representation}). For fixed operator settings, each worst case active power injection is zero or its largest feasible value (Proposition~\ref{prop:owner-best-response-fixed-kappa}). When no bus voltage crosses its nominal value, the optimal ratio cancels the corresponding injection's contribution and yields a closed form strategy (Theorem~\ref{thm:closed-form-minimax}). Third, we characterize how DER ratings and power factor restrictions affect regulation performance. Numerical tests check the cancellation computation, solve the complete minimax problem directly at a representative DER rating, and compare the linear voltage predictions with nonlinear AC power flow. The analytical settings resemble empirical power factor recommendations \cite{osti_1431468,rylander_methods_2016}, providing a game theoretic characterization.

Section~\ref{sec:prelim} develops the network model and robust power factor design. Section~\ref{sec:analytical-results} derives the minimax structure and analytical strategy. Section~\ref{sec:numerical-results} verifies the results, and Section~\ref{sec:discussion} concludes the paper.

\fi

\section{Problem Description}
\label{sec:prelim}
This section defines the network and constant power factor models and then formulates the operator's robust minimax design. Under the security interpretation, the same nested problem is a zero-sum Stackelberg game.
\subsection{Network and Device Model}
\label{sec:prelim-network-model}
We first define the network operating point, voltage approximation, constant power factor model, and regulation target.

\subsubsection{Electric power network model}
We consider an electric network $(\setN\cup\{0\},\setE)$, where $\setN=\{1,\dots,n\}$ contains the $n$ PQ nodes whose voltages depend on the nodal power injections. Node $0$ is the reference node and often represents a substation that maintains a constant voltage magnitude. A superscript $(\cdot)^0$ denotes a quantity at the initial operating point. Let $\vp^{\sf g},\vq^{\sf g}\in\R^n$ and $\vp^{\sf d},\vq^{\sf d}\in\R^n$ be the generated and demanded active and reactive power. Initial net injections are
\[
\vp^0\triangleq\vp^{\sf g}-\vp^{\sf d},
\qquad
\vq^0\triangleq\vq^{\sf g}-\vq^{\sf d}.
\]
The initial voltage phasors are $\vv^0\circ\exp(j\vtheta^0)$, where $\vv^0\in\R^n$ contains the voltage magnitudes, $\vtheta^0\in(-\pi,\pi]^n$ contains the voltage angles, and $j\triangleq\sqrt{-1}$.

\subsubsection{Approximation of the power flow equations}
Let $\bar{s}_i>0$ be the apparent power rating of the modeled device or aggregation at node $i$. We use $p_i$ and $q_i$ for its active and reactive output normalized by $\bar{s}_i$. The physical device output is therefore $(\bar{s}_i p_i,\bar{s}_i q_i)$. We take this controllable contribution to be zero at the initial operating point, so the total net injections at node $i$ after the device acts are $p_i^0+\bar{s}_i p_i$ and $q_i^0+\bar{s}_i q_i$. Thus, the power factor applies to the final output of the modeled device rather than the total net injection at the bus.

A linear approximation of the voltage magnitudes about the initial operating point is the map $\vv:\R^n\times\R^n\to\R^n$ given~by
\begin{equation}
\label{eq:linear-voltage-model}
\vv(\vp,\vq)\triangleq\vv^0+\mR\vp+\mX\vq,
\end{equation}
where $\mR,\mX\in\R^{n\times n}$ are the active and reactive power voltage sensitivity matrices evaluated at $(\vp^0,\vq^0)$, with the device ratings $\bar{s}_i$ absorbed into their corresponding columns.

\subsubsection{Implicit Representation of Reactive Power}
\label{sec:prelim-implicit}
We model DERs that export nonnegative active power. Under the rating normalization above, a feasible device output satisfies $p_i\geq0$ and $p_i^2+q_i^2\leq1$.

Constant power factor control assigns each device a power factor magnitude $\alpha_i\in(0,1]$ and a reactive power direction $\xi_i\in\{-1,0,1\}$. With the net injection convention, $\xi_i=1$ denotes reactive power injection and $\xi_i=-1$ denotes reactive power absorption. At unity power factor the direction is immaterial, and we set $\xi_i=0$ to avoid duplicate representations. For any realized output with $p_i>0$, consistency with the installed settings means
\begin{equation}
\label{eq:prelim-pfactor-def}
    \frac{p_i}{\sqrt{p_i^2+q_i^2}}=\alpha_i,
    \qquad \forall i\in\setN\text{ with }p_i>0,
\end{equation}
and
\begin{equation}
\label{eq:prelim-lead-lag}
    \sgn(q_i)=\xi_i,
    \qquad\quad\; \forall i\in\setN\text{ with }p_i>0.
\end{equation}
Equations~\eqref{eq:prelim-pfactor-def} and \eqref{eq:prelim-lead-lag} describe the realized output when $p_i>0$. At $(p_i,q_i)=(0,0)$, the realized power factor is undefined, but the installed setting $(\alpha_i,\xi_i)$ remains well defined.

A DER operating in constant power factor mode supplies reactive power in fixed proportion to its active power. The following lemma establishes this relationship.

\begin{lemma}[Reactive power representation]
\label{lemma:implicit-representation}
Fix $\vp\in\R_{\geq0}^n$ and $(\valpha,\vxi)\in(0,1]^n\times\{-1,0,1\}^n$ such that for every $i\in\setN$, $\xi_i=0$ if and only if $\alpha_i=1$. Define
\begin{subequations}
\label{eq:prelim:dqdp-def}
\begin{equation}
    \mK(\valpha,\vxi)\triangleq \diag(\vk(\valpha,\vxi)),
\end{equation}
\begin{equation}    
    k(\alpha_i,\xi_i)\triangleq \frac{\xi_i}{\alpha_i}\sqrt{1-\alpha_i^2},
    \qquad \forall i\in\setN.
\end{equation}
\end{subequations}
Let
\begin{equation}
\label{eq:implicit-reactive-representation}
    \vq\triangleq\mK(\valpha,\vxi)\vp.
\end{equation}
Then $p_i^2+q_i^2=p_i^2/\alpha_i^2$ for every $i\in\setN$. If $p_i>0$, the output satisfies \eqref{eq:prelim-pfactor-def} and \eqref{eq:prelim-lead-lag}. If $p_i=0$, then $q_i=0$.
\end{lemma}

\begin{proof}
\proofpointer{apdx:proof-lemma-implicit}{Substitution gives $p_i^2+q_i^2=p_i^2/\alpha_i^2$. The power factor and sign identities follow when $p_i>0$, while $p_i=0$ directly gives $q_i=0$.}
\end{proof}

The diagonal entries of $\mK(\valpha,\vxi)$ are the installed proportionality coefficients between reactive and active power. Related models also express reactive power as a parameterized function of active power \cite{samadi_coordinated_2014} or power factor settings \cite{fernandez_implicit_2022}.

Combining Lemma~\ref{lemma:implicit-representation} with \eqref{eq:linear-voltage-model} gives
\begin{equation}
\label{eq:implicit-vmag-approximation}
    \vv(\vp,(\valpha,\vxi)) \triangleq \vv^0 + \mS(\valpha,\vxi) \vp,
\end{equation}
where the sensitivity matrix associated with the installed settings is
\begin{equation}
\label{eq:prelim:smatrix-def}
    \mS(\valpha,\vxi)  \triangleq \DVDP + \DVDQ \mK(\valpha,\vxi).
\end{equation}
The approximation \eqref{eq:implicit-vmag-approximation} therefore includes the reactive power that accompanies each realized active power injection under constant power factor control.

\subsubsection{Nominal voltages and voltage regulation}
\label{sec:prelim:nominal-voltages}
The operator aims to keep the voltage magnitudes close to the nominal profile $\vvnom\in\R^n$. Tracking nominal voltage penalizes both upward and downward deviations and preserves operating margin to the voltage limits. Let $\vvbar,\vvubar\in\R^n$ be the lower and upper voltage magnitude bounds. We define
\[
\vmnom\triangleq\vvnom-\vv^0,
\qquad
\vmbar\triangleq\vvbar-\vv^0,
\qquad
\vmubar\triangleq\vvubar-\vv^0
\]
as the nominal, lower, and upper voltage offsets from the initial operating point.

\begin{assumption}[Bounded power factors]
\label{assum:bounded-power factors}
The devices can operate at power factors $\valpha$ satisfying $\vzero<\valphamin\leq\valpha\leq\vone$.
\end{assumption}

The lower bound $\valphamin$ represents the minimum realizable power factor at each node. It is typically determined by device capability or technical standards.

\begin{assumption}[Feasible reference voltages]
\label{assum:no-overvoltages}
The initial and nominal voltage profiles satisfy the voltage bounds. Equivalently, $\vmbar\leq\vzero\leq\vmubar$ and $\vmbar\leq\vmnom\leq\vmubar$.
\end{assumption}

The zero DER injection leaves $\vv=\vv^0$, so this assumption provides a feasible reference point before the uncertain DER injections are realized.

\subsection{Power Factor Minimax Formulation}
\label{sec:prelim:game-formulation}

The operator first installs one constant power factor setting for each DER. The active power realization then occurs within the device limits, and the reactive power is determined by the installed setting. The operator chooses the settings to minimize the largest aggregate voltage deviation from nominal that can occur afterward.

\subsubsection{Feasible decisions}
\label{sec:prelim:action-set-formulation}

Before a power factor setting is installed, the normalized nonnegative active power set is the reference box
\begin{equation}
\label{eq:prelim:injector-action-set}
\setP\triangleq[0,1]^n.
\end{equation}
The upper bound follows from the unit apparent power rating because $0\leq p_i\leq\sqrt{p_i^2+q_i^2}\leq1$. The set $\setP$ is a reference box used to define safe operator settings. The active power player's actual feasible set is defined after the operator acts.

The operator action set is
\begin{equation}
\label{eq:prelim:controller-action-set}
\begin{split}
\setA \triangleq \Big\{ &(\valpha,\vxi) \in (0,1]^n \times \{-1,0,1\}^n :
\valphamin \le \valpha \le \vone, \\
&\xi_i \in \{\pm 1\}\ \text{if } \alpha_i<1,\ \xi_i = 0\ \text{if } \alpha_i=1,\ \forall i\in\setN,\\
&\vmbar \le \mS(\valpha,\vxi)\vp' \le \vmubar,\ \forall \vp' \in \setP \Big\}.
\end{split}
\end{equation}
We require an installed setting to respect the voltage limits throughout the reference set \eqref{eq:prelim:injector-action-set}. This conservative safety requirement prevents the operator from improving the objective by selecting settings that permit unacceptable voltages. We assume that $\setA$ is nonempty.

After the operator selects $(\valpha,\vxi)\in\setA$, the apparent power rating restricts the active power realization to
\begin{equation}
\label{eq:prelim:owner-feasible-set-conditioned}
\begin{aligned}
\setP(\valpha,\vxi)
&\triangleq
\left\{
\vp\in\setP :
(1+k(\alpha_i,\xi_i)^2)p_i^2 \le 1,\: \forall i\in\setN
\right\}\\
&=\prod_{i=1}^n[0,\alpha_i].
\end{aligned}
\end{equation}
The equality follows from $1+k(\alpha_i,\xi_i)^2=1/\alpha_i^2$. Thus, lowering the installed power factor reserves more capacity for reactive power and reduces the largest feasible active power injection.

\begin{remark}[Threat model]
\label{rem:threat-model}
Under the cybersecurity interpretation, a compromised aggregator or controller observes the installed power factor settings and then selects $\vp\in\setP(\valpha,\vxi)$. The adversary can manipulate active power references through an aggregator platform, device energy management systems, or coordinated load altering attacks \cite{mohsenianrad_load_altering_2011,soltan_blackiot_2018}. We assume that the reactive power parameters remain under the Area EPS (Electric Power System) operator's control through the standardized interface and certified grid support functions \cite[Cl.~5, 10.6, 11]{noauthor_ieee_2018}.

The requirement $\vq=\mK(\valpha,\vxi)\vp$ characterizes the \textit{stealthy} attack class. Overriding reactive power would violate the installed setting and could be identified from monitoring data after the realization occurs \cite[Cl.~10]{noauthor_ieee_2018}. Thus, an intelligent malicious actor would inject so that they comply with their set power factor to avoid exposure. This restriction is analogous to false data injection attacks that remain consistent with a detector \cite{liu_false_2011}.

The vector $\vp$ may collect outputs from multiple devices. Its joint maximization represents a worst-case uncertainty set and does not require the device owners to coordinate. A compromised aggregator supplies the strategic interpretation when one actor controls the vector.
\end{remark}

The conditional set $\setP(\valpha,\vxi)$ makes the active power limits depend on the installed setting.

\subsubsection{Aggregate voltage deviation}
\label{sec:prelim:objective-formulation}

\begin{definition}[Aggregate voltage deviation]
\label{def:overvoltage-margin}
For any $(\valpha,\vxi)\in\setA$ and $\vp\in\setP(\valpha,\vxi)$, the aggregate voltage deviation from nominal is
\begin{subequations}
    \label{eq:def:margin-l1-norm}
\begin{align}
    w(\vp,(\valpha,\vxi)) &\triangleq \norm{\vvnom - \vv(\vp,(\valpha,\vxi))}_1\\
    &=\norm{\vmnom - \mS(\valpha,\vxi)\vp}_1.
\end{align}
\end{subequations}
\end{definition}

The $\ell_1$ norm assigns equal weight to voltage deviations at the participating nodes. At $\vp=\vzero$, Lemma~\ref{lemma:implicit-representation} gives $\vq=\vzero$, so the operator cannot alter voltage through a constant power factor setting alone. A device may be physically capable of supplying reactive power at zero active power under another control mode, such as constant reactive power or volt-var control, but that independent reactive power action is not available in the constant power factor mode studied here.

\subsubsection{Robust minimax formulation}
\label{sec:prelim:game-solution-concept}

The robust power factor design is
\begin{equation}
\label{eq:full-network-minimax-game}
    \min_{(\valpha,\vxi)\in\setA}\;
    \max_{\vp\in\setP(\valpha,\vxi)}
    \; w(\vp,(\valpha,\vxi)).
\end{equation}
The operator first selects $(\valpha,\vxi)$. The inner problem then selects the worst feasible active power realization from the set \eqref{eq:prelim:owner-feasible-set-conditioned}, which depends on the installed settings. Under the robust interpretation, this maximization represents uncertainty and requires no strategic actor. Under the security interpretation, \eqref{eq:full-network-minimax-game} is a sequential zero-sum Stackelberg game in which the active power player observes the installed settings before choosing $\vp$.

\section{Analytical Results}
\label{sec:analytical-results}

In these results, we first represent the power factor magnitude and direction by the continuous coefficient vector $\vkappa$. We then decompose the inner worst-case objective according to the signs of the voltage residual. On the subset where no feasible injection changes the nominal sign pattern, the decomposition yields a robust minimax setting in closed form.

This analysis uses the linear voltage approximation \eqref{eq:implicit-vmag-approximation} centered at the initial operating point. The results need not hold exactly for the nonlinear AC power flow equations, whose agreement with the linear predictions is evaluated in Section~\ref{sec:numerical-results}.

\subsection{Continuous Reformulation}
\label{sec:analy:game-transformation}

The action set $\setA$ in \eqref{eq:prelim:controller-action-set} includes a discrete choice for the direction of reactive power. The coefficient $\kappa_i$ in \eqref{eq:prelim:dqdp-def} combines this direction with the power factor magnitude. Under Assumption~\ref{assum:bounded-power factors}, its range is the box
\begin{equation}
\label{eq:analy:k-range-set}
\Range(\vk)
=
\left\{
\vkappa\in\R^n:
\vk(\valphamin,-\vct{1})
\leq\vkappa\leq
\vk(\valphamin,\vct{1})
\right\}.
\end{equation}
Define $\mStilde(\vkappa)\triangleq\mR+\mX\diag(\vkappa)$. For a fixed $\vkappa$, the feasible active power set is
\begin{equation}
\label{eq:owner-feasible-set}
\setP(\vkappa)
\triangleq
\prod_{i=1}^n
\left[0,\frac{1}{\sqrt{1+\kappa_i^2}}\right].
\end{equation}
The corresponding operator action set is
\begin{equation}
\label{eq:analy:K-set-definition}
\setK
\triangleq
\left\{
\vkappa\in\Range(\vk):
\vmbar\leq\mStilde(\vkappa)\vp'
\leq\vmubar
\quad\forall\vp'\in\setP
\right\}.
\end{equation}
For $\vkappa\in\setK$ and $\vp\in\setP(\vkappa)$, the lifted objective is
\begin{equation}
\label{eq:K-l1-lifted-margin}
\wtilde(\vp,\vkappa)
\triangleq
\norm{\vmnom-\mStilde(\vkappa)\vp}_1.
\end{equation}

The scalar $\kappa_i$ is the installed coefficient in $q_i=\kappa_i p_i$. When $p_i>0$, it also equals $q_i/p_i$. This interpretation does not require a realized power factor at zero output.

\begin{lemma}[Equivalent continuous parameterization]
\label{lemma:k-transform}
The map $(\valpha,\vxi)\mapsto\vk(\valpha,\vxi)$ is a bijection from $\setA$ to $\setK$. Its inverse is
\begin{equation}
\label{eq:apdx:k-xi-inv}
\alpha_i(\kappa_i)
\triangleq
\frac{1}{\sqrt{1+\kappa_i^2}},
\quad
\xi_i(\kappa_i)
\triangleq
\sgn(\kappa_i),
\quad i\in\setN,
\end{equation}
where $\sgn(0)=0$. If $\vkappa=\vk(\valpha,\vxi)$, then
\[
\setP(\valpha,\vxi)=\setP(\vkappa)
\quad\text{and}\quad
w(\vp,(\valpha,\vxi))=\wtilde(\vp,\vkappa)
\]
for every feasible $\vp$.
\end{lemma}

\begin{proof}
\proofpointer{apdx:proof-lemma-k-transform}{The inverse formulas recover the power factor magnitude and direction uniquely. Substitution gives the same reactive output, feasible active power set, and objective.}
\end{proof}

Thus, the partly discrete description $(\valpha,\vxi)$ can be replaced by $\vkappa$ without changing the model.

Define the exact inner worst-case value
\begin{equation}
\label{eq:worst-case-owner-value}
I(\vkappa)
\triangleq
\max_{\vp\in\setP(\vkappa)}
\wtilde(\vp,\vkappa).
\end{equation}
Lemma~\ref{lemma:k-transform} makes \eqref{eq:full-network-minimax-game} equivalent to
\begin{equation}
\label{eq:exact-single-stage-game-in-kappa}
\min_{\vkappa\in\setK} I(\vkappa).
\end{equation}

\begin{lemma}[Existence of a minimax solution]
\label{lemma:compact-convex action sets}
The set $\setK$ is nonempty, compact, and convex. For every $\vkappa\in\setK$, the set $\setP(\vkappa)$ is nonempty, compact, and convex. The function $I$ is continuous on $\setK$, so \eqref{eq:exact-single-stage-game-in-kappa} attains a minimum.
\end{lemma}

\begin{proof}
\proofpointer{apdx:proof-of-compact-cvx-action-sets}{The set $\setK$ is a bounded box intersected with closed affine halfspaces. Parameterizing each conditional active power box by $[0,1]^n$ makes continuity of $I$ follow from the maximum theorem.}
\end{proof}

\subsection{Exact Orthant Reduction}

\label{sec:analy:greedy-conservative}

For $\vs\in\{-1,0,1\}^n$ and $\vkappa\in\setK$, define
\begin{equation}
\label{eq:orthant-region-definition}
\mathcal R_{\vs}(\vkappa)
\triangleq
\left\{
\vp\in\setP(\vkappa):
\sgn\bigl(\vmnom-\mStilde(\vkappa)\vp\bigr)=\vs
\right\},
\end{equation}
where $\sgn$ is applied componentwise with $\sgn(0)=0$. These regions form an exact disjoint partition of $\setP(\vkappa)$.

\begin{proposition}[Exact orthant decomposition]
\label{prop:orthant-wise-affine-representation}
For every $\vkappa\in\setK$,
\begin{equation}
\label{eq:exact-orthant-decomposition}
\begin{aligned}
I(\vkappa)
=
\max_{\vs\in\{-1,0,1\}^n}
&\max_{\vp\in\mathcal R_{\vs}(\vkappa)}
\Bigl\{
\vs^\top\vmnom
{}\\
&+ \left(
-\mR^\top\vs-\diag(\vkappa)\mX^\top\vs
\right)^\top\vp
\Bigr\},
\end{aligned}
\end{equation}
where empty regions are omitted.
\end{proposition}

\begin{proof}
If $\vp\in\mathcal R_{\vs}(\vkappa)$, then
\[
\wtilde(\vp,\vkappa)
=
\vs^\top\bigl(\vmnom-\mStilde(\vkappa)\vp\bigr).
\]
Substituting $\mStilde(\vkappa)=\mR+\mX\diag(\vkappa)$ gives the affine expression in \eqref{eq:exact-orthant-decomposition}. Maximizing over the exact partition gives $I(\vkappa)$.
\end{proof}

The sign pattern
\(
\vs^\bullet\triangleq\sgn(\vmnom)
\)
is the sign of the voltage residual at zero DER output. Since $\vmnom=\vvnom-\vv^0$, its entries show whether initial voltages are below, above, or at nominal. For fixed $\vkappa$, $\mathcal R_{\vs^\bullet}(\vkappa)$ contains injections that leave every nonnominal voltage strictly on its initial side. A voltage reaching nominal has a zero entry in its exact sign pattern.

\begin{definition}[Nominal orthant subset]
\label{def:nominal orthant-subset}
Define
\[
\setK^\bullet\triangleq
\left\{\vkappa\in\setK: \setP(\vkappa) = 
\overline{\mathcal R_{\vs^\bullet}(\vkappa)}
\right\}.
\]
\end{definition}

The closure in $\R^n$ allows voltages to reach nominal without crossing to the opposite side. For each $j$ with $s_j^\bullet\neq0$, the condition is equivalent to
\begin{equation}
\label{eq:nominal-orthant-row-condition}
\sum_{i=1}^n
\frac{\max\!\left\{s_j^\bullet[\mStilde(\vkappa)]_{ji},0\right\}}
{\sqrt{1+\kappa_i^2}}
\leq \left|m_j^\bullet\right|.
\end{equation}
If $s_j^\bullet=0$, the corresponding row of $\mStilde(\vkappa)$ must be zero. Smaller ratings, weaker voltage sensitivity toward nominal, and greater initial distance favor membership; voltage limit headroom separately determines whether $\vkappa\in\setK$. For uniform apparent power ratings $\bar{s}_i=\bar{s}$ and every $m_j^\bullet\neq0$, the threshold for fixed $\vkappa$ is
\begin{equation}
\label{eq:uniform-rating-nominal-orthant-threshold}
\bar{s}\leq
\min_{j\in\setN}
\frac{\left|m_j^\bullet\right|}
{\displaystyle\sum_{i=1}^n
\frac{\max\!\left\{s_j^\bullet([\mR]_{ji}+[\mX]_{ji}\kappa_i),0\right\}}
{\sqrt{1+\kappa_i^2}}},
\end{equation}
where $\mR$ and $\mX$ are evaluated at unit ratings, and a zero denominator imposes no restriction. For $\vkappa\in\setK^\bullet$, the $\ell_1$ objective is affine over the entire feasible set, so the inner maximization is separable. Define
\begin{equation}
\label{eq:nominal-sigma-omega}
\vsigma\triangleq-\mR^\top\vs^\bullet,
\qquad
\vomega\triangleq-\mX^\top\vs^\bullet.
\end{equation}

\begin{proposition}[Nominal orthant value and worst-case realizations]
\label{prop:owner-best-response-fixed-kappa}
For every $\vkappa\in\setK^\bullet$, the exact inner value is
\begin{equation}
\label{eq:nominal-orthant-value}
I(\vkappa)
=
\norm{\vmnom}_1
+
\sum_{i=1}^n
\frac{[\sigma_i+\omega_i\kappa_i]_+}
{\sqrt{1+\kappa_i^2}}.
\end{equation}
The worst-case active power realizations are characterized componentwise by
\[
p_i^{\mathrm{wc}}=
\begin{cases}
\dfrac{1}{\sqrt{1+\kappa_i^2}},
& \sigma_i+\omega_i\kappa_i>0,\\[2mm]
0,
& \sigma_i+\omega_i\kappa_i<0,\\
\text{\rm any value in }
\left[0,\dfrac{1}{\sqrt{1+\kappa_i^2}}\right],
& \sigma_i+\omega_i\kappa_i=0.
\end{cases}
\]
\end{proposition}

\begin{proof}
For $\vkappa\in\setK^\bullet$, Definition~\ref{def:nominal orthant-subset} gives
\[
\norm{\vmnom-\mStilde(\vkappa)\vp}_1
=
(\vs^\bullet)^\top
\bigl(\vmnom-\mStilde(\vkappa)\vp\bigr)
\]
for every $\vp\in\setP(\vkappa)$, including points at which a voltage reaches nominal. Substitution gives
\[
\norm{\vmnom}_1
+
\sum_{i=1}^n
(\sigma_i+\omega_i\kappa_i)p_i.
\]
The feasible set in \eqref{eq:owner-feasible-set} is a Cartesian product. Maximizing each linear term over its interval gives the stated realizations and value.
\end{proof}
For a fixed power factor setting, $\sigma_i+\omega_i\kappa_i$ is the contribution of one unit of active power at node $i$ to the affine expression for the aggregate voltage deviation. It accounts for the reactive power that accompanies the active power. A positive coefficient makes the largest feasible injection worst for the operator, while a negative coefficient makes zero injection worst.

\subsection{Cancellation Settings and Feasibility}
\label{sec:analytical-minimax-conditions}

We call a vector a \emph{cancellation setting} if it satisfies
\(
\sigma_i+\omega_i\kappa_i=0
\)
for every $i\in\setN$. Such a vector cancels every injection coefficient in \eqref{eq:nominal-orthant-value}. The componentwise cancellation conditions alone do not guarantee feasibility. The resulting vector must satisfy the range and robust voltage constraints in \eqref{eq:analy:K-set-definition} and the nominal orthant condition in Definition~\ref{def:nominal orthant-subset}.

If $\omega_i=0$ and $\sigma_i\leq0$, the $i$th summand in \eqref{eq:nominal-orthant-value} is zero for every feasible $\kappa_i$, so this coordinate can be chosen to help satisfy the coupled constraints. If the coefficients at all other coordinates are canceled and the resulting vector belongs to $\setK^\bullet$, the same minimax value $\norm{\vmnom}_1$ follows. If $\omega_i=0$ and $\sigma_i>0$, the summand is $\sigma_i/\sqrt{1+\kappa_i^2}$ and cannot be reduced to zero by cancellation. In this case, the operator must minimize the exact value \eqref{eq:worst-case-owner-value} over $\setK$.

\begin{theorem}[Robust minimax solution in closed form]
\label{thm:closed-form-minimax}
Assume $\omega_i \neq 0$ for all $i \in \setN$. Define
\[
\kappa_i^* \triangleq -\frac{\sigma_i}{\omega_i},
\qquad i \in \setN.
\]
Suppose $\vkappa^* \in \setK^\bullet$. Then $\vkappa^*$ is a global solution of \eqref{eq:exact-single-stage-game-in-kappa}, and
\[
\min_{\vkappa\in\setK}I(\vkappa)
=I(\vkappa^*)
=\norm{\vmnom}_1.
\]
Its original power factor representation is
\begin{subequations}
\label{eq:minimax-pf}
\begin{align}
\alpha_i^*
&=
\frac{1}{\sqrt{1+(\kappa_i^*)^2}}
=
\frac{1}{\sqrt{1+(\sigma_i/\omega_i)^2}},
\\
\xi_i^*
&=
\sgn(\kappa_i^*)
=
-\sgn(\sigma_i\omega_i),
\qquad i\in\setN.
\end{align}
\end{subequations}
Every $\vp^*\in\setP(\vkappa^*)$ is a worst-case realization. In particular, the choice
\[
p_i^*=\frac{1}{\sqrt{1+(\kappa_i^*)^2}}
\]
with $q_i^*=\kappa_i^*p_i^*$ saturates every normalized apparent power constraint, i.e., $(p_i^*)^2+(q_i^*)^2=1$ for all $i$.
\end{theorem}

\begin{proof}
For every $i$,
\[
\sigma_i+\omega_i\kappa_i^*
=
0,
\]
so Proposition~\ref{prop:owner-best-response-fixed-kappa} gives $I(\vkappa^*)=\norm{\vmnom}_1$. For every $\vkappa\in\setK$, the feasible realization $\vp=\vzero$ gives
\[
\wtilde(\vzero,\vkappa)=\norm{\vmnom}_1.
\]
Hence $I(\vkappa)\geq\norm{\vmnom}_1$ for every $\vkappa\in\setK$, and $\vkappa^*$ attains this global lower bound. The formulas for $(\valpha^*,\vxi^*)$ follow from Lemma~\ref{lemma:k-transform}. Since every coefficient in Proposition~\ref{prop:owner-best-response-fixed-kappa} is zero at $\vkappa^*$, every feasible $\vp^*$ attains the same inner value.
\end{proof}

If the cancellation vector does not belong to $\setK^\bullet$, the theorem does not necessarily solve the exact minimax problem. It still cancels the nominal affine branch, but another voltage sign pattern may determine $I(\vkappa)$. In that case, the exact objective \eqref{eq:worst-case-owner-value} must be minimized over $\setK$.

The operator cannot control the active power injections, but it controls the reactive power that accompanies each unit of active power. The setting $\kappa_i^*=-\sigma_i/\omega_i$ cancels the contribution of the $i$th injection to the nominal affine branch of the aggregate voltage deviation. When $\vkappa^*\in\setK^\bullet$, no feasible injection raises the aggregate deviation above the fixed initial offset $\norm{\vmnom}_1$. This is the sense in which the minimax power factors neutralize worst-case active power injections.

The minimax power factors \eqref{eq:minimax-pf} depend explicitly on the topology-dependent quantities $\vsigma$ and $\vomega$. 
Moreover, \eqref{eq:minimax-pf} is identical to the recommended power factor settings prescribed by \cite{osti_1431468,rylander_methods_2016}. Thus, these empirical rules can be interpreted as robust reactive power responses to worst-case active power perturbations.

\section{Numerical Results}
\label{sec:numerical-results}

The four test networks are the 37-bus Hawaii and 200-bus Illinois synthetic grids from Texas A\&M \cite{birchfield_grid_2017}, the 73-bus RTS-GMLC \cite{barrows_ieee_2019}, and the IEEE 118-bus case from MATPOWER \cite{zimmerman_matpower_2011}. In every table, $\abs{PQ}$ counts participating PQ nodes. Thus, for example, RTS-GMLC has 40 participating nodes among its 73 buses. The voltage sensitivities are derivatives of the nonlinear AC power flow equations at each initial operating point, while the optimization experiments use the resulting linear model. We model each optimization problem described below in JuMP \cite{jump} and solve it with Gurobi \cite{gurobi}.

\subsection{Numerical Check of the Cancellation Setting}
Table~\ref{tab:minimax-pfs} compares the analytical cancellation setting $\vkappa^*$ from Theorem~\ref{thm:closed-form-minimax} with a numerical implementation check that does not use the ratio $-\sigma_i/\omega_i$. We compute
\begin{equation}
\label{eq:numerical-cancellation-program}
\vkappa^{\mathrm{num}}
\in
\operatorname*{arg\,min}_{\vkappa\in\Range(\vk)}
\sum_{i=1}^n(\sigma_i+\omega_i\kappa_i)^2,
\end{equation}
where $\Range(\vk)$ is the lifted box in \eqref{eq:analy:k-range-set}. This is the program used to produce Table~\ref{tab:minimax-pfs}. Its Hessian is $2\diag(\omega_1^2,\ldots,\omega_n^2)$, so it is a convex quadratic least squares problem with box constraints. The program checks the numerical model of the cancellation equations and the inverse map to the power factor variables.

A direct numerical solution of \eqref{eq:exact-single-stage-game-in-kappa} must account for every feasible voltage sign pattern. The operator setting also changes both the voltage response and the owner's active power bounds. Across all four networks, $\vkappa^{\mathrm{num}}$ and its power factor representation agree with the analytical setting at solver precision. This agreement verifies the cancellation computation but is not independent evidence of global minimax optimality. When $\vkappa^*\in\setK^\bullet$, Theorem~\ref{thm:closed-form-minimax} proves global optimality. Section~\ref{sec:numres:rating-feasibility} separately evaluates the fixed cancellation setting and solves the complete minimax problem directly.

\begin{table}[t]
    \renewcommand{\arraystretch}{1.25}
    \caption{The numerical cancellation check \eqref{eq:numerical-cancellation-program} recovers the analytical settings at solver precision.}
    \centering
    \begin{tabular}{|c|c|c|c|c|}
    \hline
\multicolumn{2}{|c|}{Test Network} & \multicolumn{2}{c|}{Relative Error}  & \multirow{2}{*}{Runtime} \\
    \cline{1-4}
         Name & $\abs{PQ}$       & $\valpha^*$         & $\vkappa^*$         &   \\
         \hline
         Hawaii 37       & 27  & $1.26 \times 10^{-10}$ & $8.71 \times 10^{-10}$ & 0.05 s  \\
         RTS-GMLC        & 40  & $1.86 \times 10^{-7}$  & $3.15 \times 10^{-6}$  & 0.002 s \\
         IEEE 118        & 64  & $1.92 \times 10^{-8}$  & $1.49 \times 10^{-7}$  & 0.002 s \\
         Illinois 200    & 162 & $2.89 \times 10^{-14}$ & $5.67 \times 10^{-13}$ & 0.012 s \\
    \hline
    \end{tabular}
    \label{tab:minimax-pfs}
\end{table}

We also apply the injections that saturate the apparent power constraints at $\vkappa^*$ to the nonlinear AC power flow equations and compare the resulting voltages with the linear approximation. When $\vkappa^*\in\setK^\bullet$, Theorem~\ref{thm:closed-form-minimax} identifies these injections as representative worst case realizations. Otherwise, they are boundary test points without that guarantee. We measure agreement using the maximum absolute voltage magnitude difference across participating nodes and the relative difference between the two aggregate deviations. At a uniform DER rating $\bar{s}=0.01$~p.u. per participating node, the first metric ranges from $5.4\times10^{-6}$ to $2.1\times10^{-3}$~p.u. across the four networks. The second ranges from $0.01\%$ to $6.1\%$. Both metrics come from evaluating the injections at the boundary of the feasible set with the linear and nonlinear models. The discrepancies increase with the uniform apparent power rating $\bar{s}$. On the Illinois 200-bus case, the maximum voltage magnitude difference reaches $0.057$~p.u. at $\bar{s}=0.05$~p.u. Thus, the linear model is accurate for small deviations from the initial operating point, while the analytical result may be less applicable to large injections under the nonlinear AC power flow model. The cancellation result fixes the aggregate deviation at its initial offset under its assumptions, but it does not require every voltage to remain at its initial value.

\subsection{Robust Objective Versus the Minimum Operator Power Factor}
\label{sec:numeres:game-value-vs-floor}

We next examine how the nominal orthant expression varies with the minimum operator power factor $\alphamin$ that bounds the lifted action set. Let $\widehat{\vkappa}$ be the cancellation setting projected onto the lifted box allowed by $\alphamin$. Define the excess above the fixed initial offset as
\[
V^\bullet(\widehat{\vkappa})
\triangleq
\sum_{i=1}^n
\frac{[\sigma_i+\omega_i\widehat{\kappa}_i]_+}
{\sqrt{1+\widehat{\kappa}_i^2}}.
\]
Fig.~\ref{fig:game-value-vs-floor} plots $\norm{\vmnom}_1+V^\bullet(\widehat{\vkappa})$ for $\alphamin\in[0.7,0.99]$. When the unconstrained cancellation setting $\vkappa^*$ is realizable within the floor, every summand vanishes and the value equals $\norm{\vmnom}_1$. As $\alphamin$ tightens and clips $\vkappa^*$, the plotted expression increases, quantifying the regulation capability forgone by restricting reactive power. The floor at which the value departs from $\norm{\vmnom}_1$ is network dependent and is governed by the smallest cancellation power factor. The value rises sharply as $\alphamin\to1$, where the operator is driven to unity power factor and retains no reactive authority.

\subsection{Comparison With a Centralized Curtailment Benchmark}
\label{sec:numres:pareto}
As an operational reference, we formulate an approximate optimal power flow model that minimizes active power curtailment. Maximizing the total normalized active power is equivalent to minimizing curtailment from the available upper bounds in $\setP$. The centralized planner selects the active power and reactive to active power ratios jointly. In particular, we solve
\begin{equation}
\label{eq:benevolent-program}
    \max_{\vp \in \setP ,\vkappa \in \setK} \vct{1}^\T \vp \quad \text{s.t.} \quad  \vmbar \leq \mStilde(\vkappa)\vp \leq \vmubar.
\end{equation}
The centralized benchmark \eqref{eq:benevolent-program} gives one planner complete knowledge of the active and reactive power decisions. Its objective differs from the voltage deviation objective in \eqref{eq:full-network-minimax-game}. It provides a conventional curtailment reference for comparing the voltage behavior that results from centralized coordination with the behavior of the robust power factor setting.

\begin{figure}[t]
    \centering
    \begin{overpic}[width=1\linewidth,keepaspectratio]{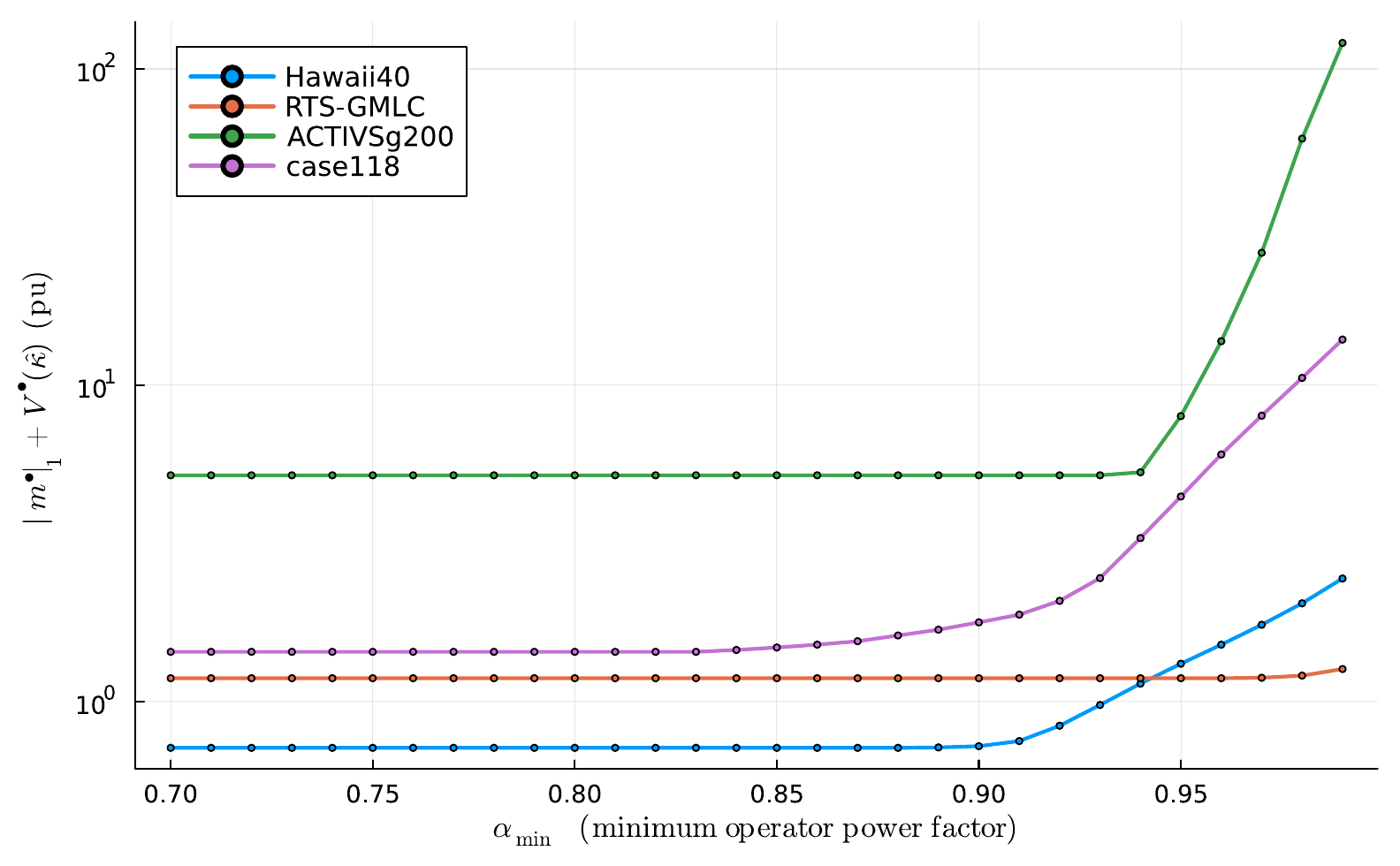}
        \put(0,22){\rotatebox{90}{\colorbox{white}{\parbox[c][1.8em][c]{0.36\linewidth}{\centering\scriptsize $\norm{\vmnom}_1+V^\bullet(\widehat{\vkappa})$ (p.u.)}}}}
        \put(28,-0.5){\colorbox{white}{\parbox[c][1.1em][c]{0.62\linewidth}{\centering\scriptsize $\alphamin$ (minimum operator power factor)}}}
    \end{overpic}
    \caption{Nominal orthant expression as a function of the minimum operator power factor $\alphamin$, where $\widehat{\vkappa}$ is the constrained cancellation setting. The value is flat at $\norm{\vmnom}_1$ while the cancellation power factors are realizable within the floor and increases once the floor clips $\vkappa^*$.}
    \label{fig:game-value-vs-floor}
\end{figure}

We compare the centralized curtailment solution with the cancellation setting
\ifextendedversion
in Fig.~\ref{fig:comp-vs-benev}
\fi
using the Hawaii and RTS-GMLC test networks. We show a representative realization that saturates each apparent power constraint at the cancellation setting. This is a worst case realization when the setting lies in $\setK^\bullet$. On Hawaii, the centralized planner injects slightly more total active power than the cancellation realization ($\vct{1}^\top\vp=26.7$ versus $25.0$~p.u.), but incurs substantially larger voltage deviations. The mean distance to nominal is $0.089$~p.u. under the centralized benchmark and $0.027$~p.u. at the cancellation setting. Indeed, the realized $\ell_1$ deviation, $27\times0.027\approx0.72$~p.u., matches the robust minimax value $\norm{\vmnom}_1=0.71$~p.u. predicted by Theorem~\ref{thm:closed-form-minimax}. On RTS-GMLC, the ordering reverses. The mean distance is $0.235$~p.u. at the cancellation setting and $0.031$~p.u. under the centralized benchmark. At full normalized rating, $\setK$ is empty under the voltage limits supplied with this case. Therefore Theorem~\ref{thm:closed-form-minimax} is not globally optimal on this instance.

\ifextendedversion
\begin{figure}
    \centering
    \includegraphics[width=1\linewidth,keepaspectratio]{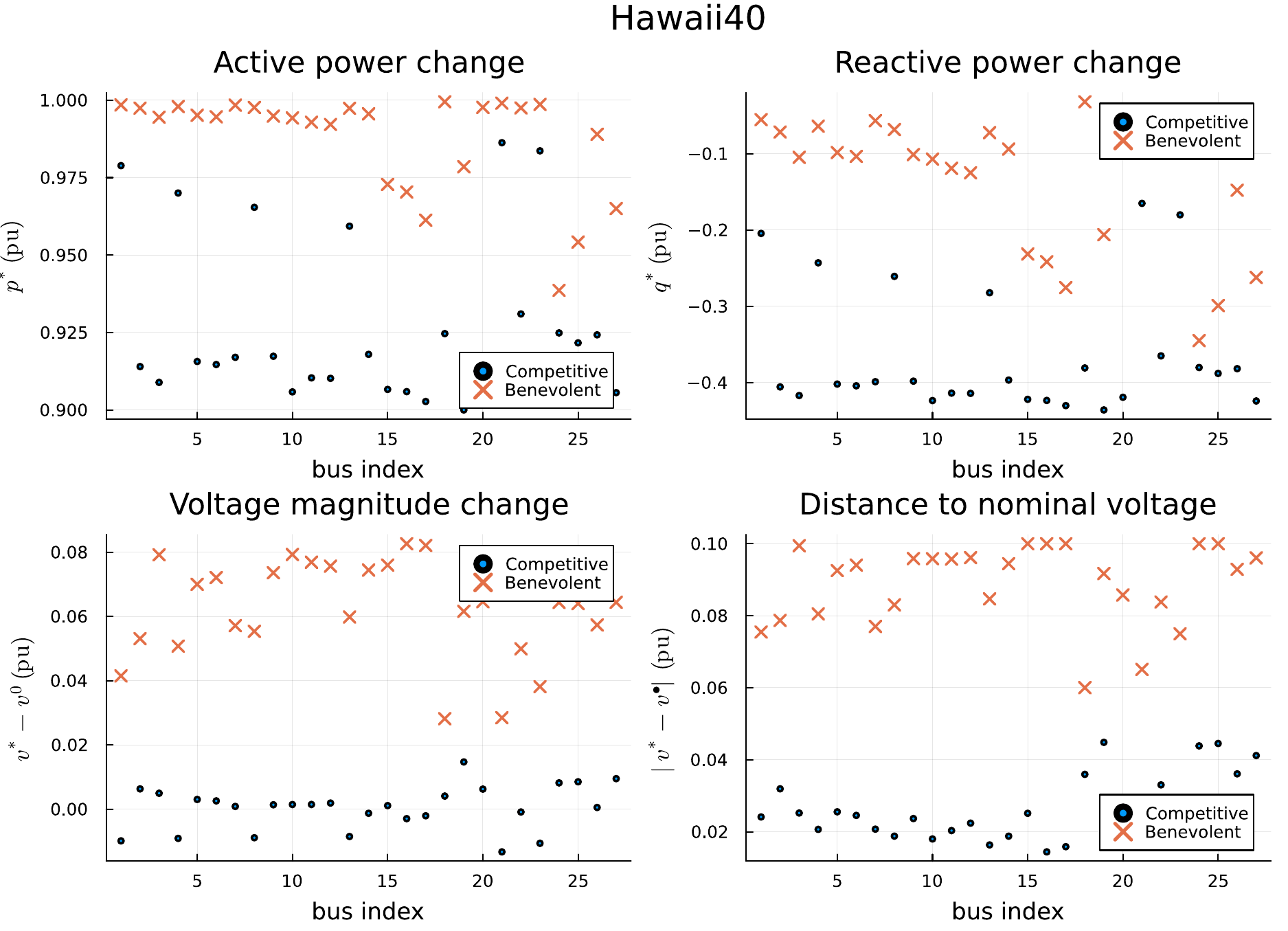}
    \includegraphics[width=1\linewidth,keepaspectratio]{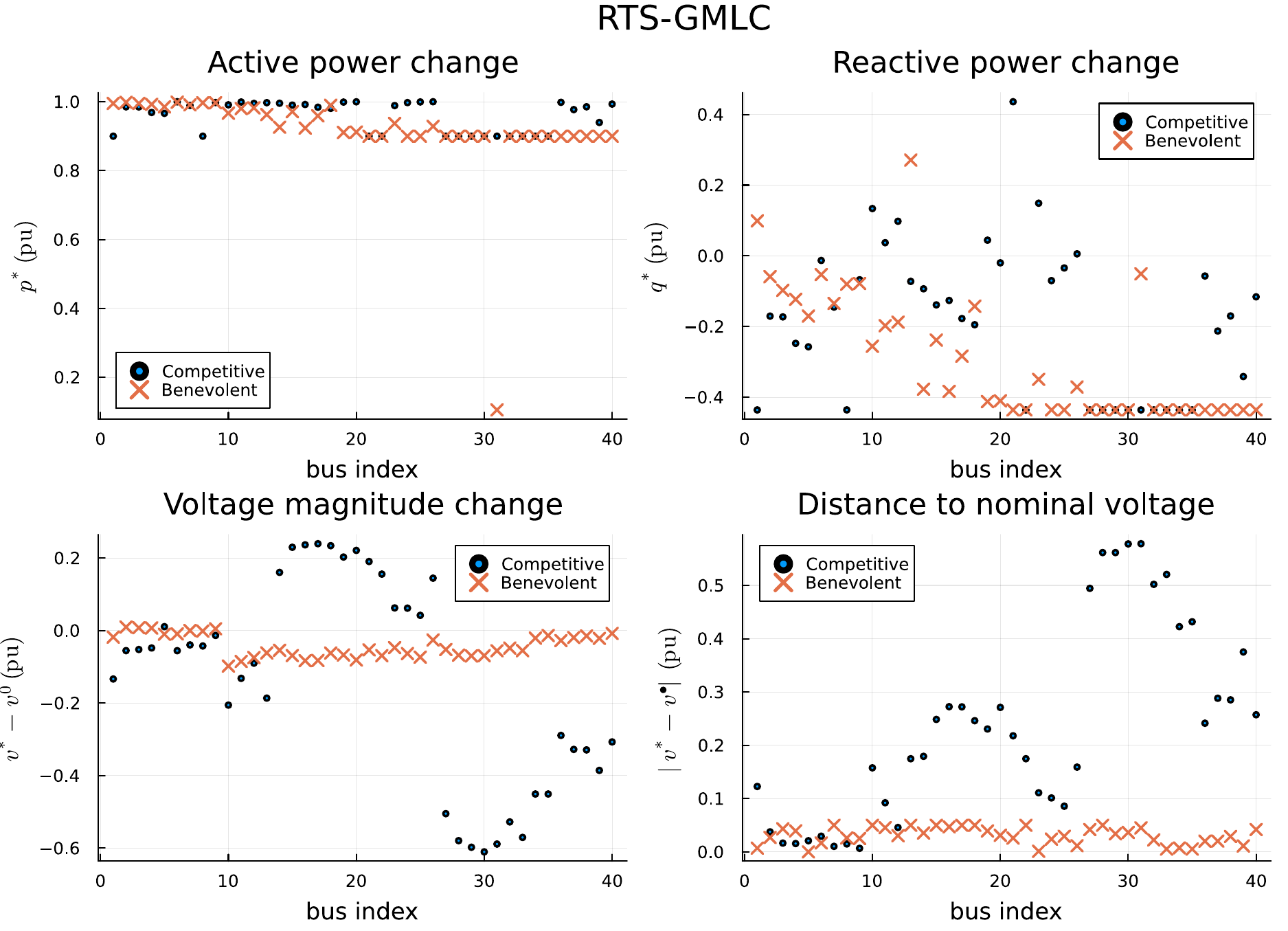}
    \caption{Comparison of a boundary realization at the cancellation setting with the centralized benchmark \eqref{eq:benevolent-program} on the Hawaii (top) and RTS-GMLC (bottom) networks. The Hawaii setting is a feasible robust policy. The RTS-GMLC setting is only a diagnostic because $\setK$ is empty at this rating. In each panel, $\vp^*$ and $\vq^*$ denote the normalized active and reactive outputs of the displayed solution, and $\vv^*$ denotes its voltage magnitudes. The vectors $\vv^0$ and $\vvnom$ are the initial and nominal voltage profiles. The apparent power constraints at each node are the unit ball in the complex plane.}
    \label{fig:comp-vs-benev}
\end{figure}
\fi

\subsection{Rating Feasibility and the Nominal Orthant}
\label{sec:numres:rating-feasibility}

On the nominal orthant subset $\setK^\bullet$ of Definition~\ref{def:nominal orthant-subset}, where the owner's entire feasible set $\setP(\vkappa)$ remains in the nominal orthant, Theorem~\ref{thm:closed-form-minimax} shows that the cancellation setting $\vkappa^*$ attains the global lower bound $\norm{\vmnom}_1$ on the worst case objective. The operator cannot reduce its worst case deviation below this value because zero injection is always feasible. No owner realization can raise the deviation above it at $\vkappa^*$. Thus, the exact value in \eqref{eq:worst-case-owner-value} satisfies $I(\vkappa^*)=\norm{\vmnom}_1$.

We first fix the operator action at $\vkappa^*$ and compute the inner maximization \eqref{eq:worst-case-owner-value} exactly as a mixed integer program. This is the inner maximization of the game. It evaluates $I(\vkappa^*)$ without solving the outer minimization over $\setK$. We begin with a uniform device rating $\bar{s}=1$~p.u. on the system base at every participating node. The ratio $I(\vkappa^*)/\norm{\vmnom}_1$ is $1.00$ for Hawaii, $2.47$ for Illinois 200, $13.05$ for RTS-GMLC, and $15.24$ for IEEE 118. A ratio above one means that an injection near the device rating can push the voltage residual across the nominal profile into another orthant. The value $\norm{\vmnom}_1$ is then only the nominal orthant prediction at the cancellation setting, while its exact inner value is strictly larger. The choice $\bar{s}=1$~p.u. is a deliberate stress test in which every participating node can inject apparent power equal to the system base. The robust model does not require this rating. This fixed action calculation alone does not solve the outer minimization or establish that $\vkappa^*$ belongs to $\setK$.

The quantity $\bar{s}_i$, introduced in Section~\ref{sec:prelim-network-model}, is the physical apparent power rating at node $i$ expressed on the system base. In practice, these ratings are far below the system base. On a 100~MVA base, $\bar{s}_i=0.05$ to $0.1$~p.u. corresponds to 5 to 10~MVA of devices at every participating node, or hundreds of MVA of aggregate DER capacity on these systems. Normalizing each device output by $\bar{s}_i$ is equivalent to scaling the corresponding columns of the voltage sensitivity matrices $\mR$ and $\mX$ by $\bar{s}_i$. This scaling leaves $(\vkappa^*,\valpha^*,\vxi^*)$ and the nominal value $\norm{\vmnom}_1$ unchanged because $\bar{s}_i$ cancels in the ratio $\vsigma/\vomega$. Table~\ref{tab:rating-sweep} reports $I(\vkappa^*)/\norm{\vmnom}_1$ as a uniform rating $\bar{s}$ is reduced. The fixed cancellation value equals the lower bound for Hawaii at all reported ratings and for Illinois 200 by $\bar{s}=0.1$. RTS-GMLC and IEEE 118 approach the lower bound with ratios $1.14$ and $1.22$ at $\bar{s}=0.05$. For RTS-GMLC, \eqref{eq:uniform-rating-nominal-orthant-threshold} gives $\bar{s}\leq0.0151$~p.u. at $\vkappa^*$, but $\setK^\bullet$ remains empty at every positive rating because three initial voltages equal their upper limits and the corresponding robust constraints are incompatible. The power factors are invariant to within $10^{-14}$ under changes in $\bar{s}$. These values describe the performance of the cancellation setting.

\begin{table}[t]
    \renewcommand{\arraystretch}{1.25}
    \caption{Worst-case ratio $I(\vkappa^*)/\norm{\vmnom}_1$ at the fixed cancellation setting versus the apparent power rating $\bar{s}$ at each node (system p.u.). A ratio of $1.00$ means that its value equals the universal lower bound.}
    \centering
    \begin{tabular}{|c|c|c|c|c|c|}
    \hline
        Network & $\bar{s}=1$ & $0.5$ & $0.2$ & $0.1$ & $0.05$ \\
    \hline
        Hawaii 37    & 1.00  & 1.00 & 1.00 & 1.00 & 1.00 \\
        RTS-GMLC     & 13.05 & 6.65 & 2.83 & 1.62 & 1.14 \\
        IEEE 118     & 15.24 & 7.65 & 3.13 & 1.77 & 1.22 \\
        Illinois 200 & 2.47  & 1.45 & 1.03 & 1.00 & 1.00 \\
    \hline
    \end{tabular}
    \label{tab:rating-sweep}
\end{table}

We next solve the complete outer problem \eqref{eq:exact-single-stage-game-in-kappa} by constraint generation over voltage sign patterns. The master problem enforces the full operator set $\setK$ and provides a global lower bound. The exact mixed integer inner problem finds a worst case sign pattern and provides an upper bound. Gurobi solves both the global nonconvex master and the inner problem. Agreement of the bounds certifies the minimax value.

For a concise comparison, define
\[
L\triangleq\norm{\vmnom}_1,
\qquad
C\triangleq I(\vkappa^*),
\qquad
V\triangleq\min_{\vkappa\in\setK}I(\vkappa).
\]
Table~\ref{tab:direct-minimax} reports the representative rating $\bar{s}=0.05$~p.u. with $\alphamin=0.7$. This power factor floor is also the lower endpoint used in Fig.~\ref{fig:game-value-vs-floor}. The cancellation value $C$ is a fixed action diagnostic even when $\vkappa^*$ is inadmissible. The direct value $V$ is reported only when $\setK$ is nonempty. 

\begin{table}[t]
    \renewcommand{\arraystretch}{1.25}
    \caption{Direct minimax comparison at $\bar{s}=0.05$~p.u. and $\alphamin=0.7$. Here, $C=I(\vkappa^*)$, $V=\min_{\vkappa\in\setK}I(\vkappa)$, and $L=\norm{\vmnom}_1$. A dash indicates that the operator set is empty.}
    \centering
    \begin{tabular}{|c|c|c|c|}
    \hline
        Network & Set relation & $C/L$ & $V/L$ \\
    \hline
        Hawaii 37    & $\vkappa^*\in\setK^\bullet$ & 1.00 & 1.00 \\
        RTS-GMLC     & $\setK=\varnothing$          & 1.14 & --   \\
        IEEE 118     & $\vkappa^*\notin\setK$      & 1.22 & 1.00 \\
        Illinois 200 & $\vkappa^*\in\setK^\bullet$ & 1.00 & 1.00 \\
    \hline
    \end{tabular}
    \label{tab:direct-minimax}
\end{table}

The global lower and upper bounds agree at solver precision for the three nonempty cases. On Hawaii and Illinois 200, the cancellation setting belongs to $\setK^\bullet$, and the direct solution confirms the value established by Theorem~\ref{thm:closed-form-minimax}. On IEEE 118, the cancellation setting is not an admissible operator action at this power factor floor. The theorem therefore does not apply to that setting. A different admissible action nevertheless attains $V=L$, which shows that the nominal offset can remain achievable outside the cancellation regime. For RTS-GMLC, no power factor setting keeps the full reference box \eqref{eq:prelim:injector-action-set} within the voltage limits supplied with the case. Thus, $\setK$ is empty, the assumed nonempty operator set does not hold, and this model instance has no minimax value. Its reported cancellation ratio is not the value of a feasible policy.

Finally, the cancellation setting can be computed well beyond these representative cases. Across networks ranging from $14$ to $2000$ buses, the relative error between \eqref{eq:minimax-pf} and the numerical implementation check remains between $10^{-7}$ and $10^{-14}$, with each instance solved in under one second.

\section{Discussion and Conclusion}
\label{sec:discussion}
Here, we briefly discuss the results and their implications.

\subsection{Commentary on Results}
\label{sec:discussion:commentary}

\subsubsection{Intuition}
The apparent power limit links the active power capacity with the reactive response available for voltage regulation. A larger power factor reserves more capacity for active power but gives the operator less reactive authority. Section~\ref{sec:numeres:game-value-vs-floor} quantifies the resulting increase in the nominal orthant objective as the minimum allowed power factor rises. The cancellation setting selects the reactive response needed to remove each active injection coefficient when the setting is feasible.

\subsubsection{Computational tractability}
\label{sec:discussion:tractability}

The numerical results in Section~\ref{sec:numerical-results} show that the minimax setting is exceptionally tractable in the nominal orthant regime. The power factors are available in closed form from \eqref{eq:minimax-pf} and require no optimization. The numerical check in \eqref{eq:numerical-cancellation-program} is a convex quadratic least squares problem that Gurobi \cite{gurobi} solves in milliseconds. It verifies the implementation and is not a substitute for a direct numerical minimax solution. The direct outer solves at $\bar{s}=0.05$ confirm the analytical value on Hawaii and Illinois 200. They also find a different optimal action on IEEE 118, where the cancellation setting is inadmissible. Outside the nominal orthant regime, computing the minimax action requires the more demanding global constraint generation procedure.

\subsubsection{Alignment with practice}
\label{sec:discussion:alignment-engineering}

According to \cite[5.2, Table 7]{noauthor_ieee_2018}, the recommended maximum reactive power injection/absorption capability of a DER for ``normal operating performance" is 44\% of the maximum magnitude complex power injection. That is, the default $\alphamin$ is given as $\alphamin = \cos(\arcsin(0.44)) = 0.898$. This provides a value based on IEEE 1547-2018 for the power factor floor $\alphamin$ in Assumption~\ref{assum:bounded-power factors}.

Additionally, empirical research has obtained equations for recommended power factor control parameters that resemble the cancellation power factors \eqref{eq:minimax-pf}. The ``recommended" settings in \cite{osti_1431468,rylander_methods_2016}, which were empirically determined, are identical to the analytical settings derived here. When the cancellation setting belongs to $\setK^\bullet$, Theorem~\ref{thm:closed-form-minimax} gives these settings a robust minimax interpretation. Across our test networks, the median cancellation power factors range from $0.92$ to $0.99$ and lie predominantly above the IEEE 1547-2018 default of $0.898$.

\subsection{Limitations and Future Work}

\subsubsection{Structural constraints for different reactive power controls}
\label{sec:discussion:limitations:q-params}
Power factor control is the default reactive power control mode in IEEE Standard 1547-2018 \cite{noauthor_ieee_2018}. Thus, the proposed power factor game provides a practical starting point for competitive analyses of inverter control. However, the operator action set we have analyzed does not represent the control curves available under other modes specified by~\cite{noauthor_ieee_2018}. Future work should consider the other known constraints of control modes such as volt-var. The minimax reactive power response obtained here may provide a target for those controls, but their direct interaction with active power requires different operator action sets.

\subsubsection{Improved models of the power flow equations}
\label{sec:additional-models}
The closed form analysis relies on a first-order voltage model, and the nonlinear validation shows that its accuracy decreases as the injection scale grows. Future work should consider other linear models \cite{huang_generalized_2021,losada_carreno_logv_2022}, convex power flow relaxations, and full nonconvex AC formulations. The orthant reduction or classical duality arguments may not apply directly, but gradient methods such as \cite{schafer_competitive_2019} may support nonlinear formulations. When the network model is unknown, voltage sensitivities can instead be estimated from measurements \cite{gupta_model-less_2022,yeh_robust_2022} and used within the present game.

\subsubsection{Development of improved strategy profiles}
\label{sec:limitations:active-power-strategy}
An important direction for future work is to enrich the active power uncertainty and attacker models. In the power factor game, all active power agents are aggregated into the inner maximization in \eqref{eq:worst-case-owner-value}. More realistic future work can consider multiple grid aggregators, distributional uncertainty, or multistage attacker models, particularly when the grid topology and measurements from Section~\ref{sec:additional-models} are uncertain.

\ifextendedversion
A further important extension is the fully compromised setting, in
which the adversary directly selects the complex power injection $s_i$,
$\abs{s_i}\le 1$, at a compromised subset of nodes $\setC \subseteq
\setN$, while the operator retains authority over the power factor
parameters only at the uncompromised nodes $\setN\setminus\setC$. The
orthant-wise machinery of Section~\ref{sec:analytical-results}
partially extends to this setting. On each orthant, the worst-case
contribution of a compromised node $i\in\setC$ becomes a linear
maximization over the half-disk $\{s_i : \abs{s_i}\le 1,\
\Re{s_i}\ge 0\}$, yielding closed form terms such as
$\sqrt{\sigma_i^2+\omega_i^2}$ (when $\sigma_i \ge 0$) that the
operator cannot influence. Characterizing the minimax solutions of the
resulting defender--attacker partition game would quantify the value of
reactive power channel integrity.
\fi

\section*{Acknowledgement}
The authors thank Santiago Grijalva for his comments on an early version of this work.

\ifextendedversion\else
\balance
\fi
\bibliographystyle{IEEEtran}

{%
\ifdefined\useoriginaltemplate
  \ifextendedversion\else
    \renewcommand{\footnotesize}{\fontsize{7pt}{7.5pt}\selectfont}%
  \fi
\fi
\footnotesize
\bibliography{refs}
}

\ifextendedversion
\appendix

\subsection{Proof of Lemma~\ref{lemma:implicit-representation}}
\label{apdx:proof-lemma-implicit}
\begin{proof}
Fix $i\in\setN$. By definition,
\[
q_i = \frac{\xi_i}{\alpha_i}\sqrt{1-\alpha_i^2}\,p_i,
\]
so
\[
q_i^2 = \frac{1-\alpha_i^2}{\alpha_i^2}\,p_i^2.
\]
Hence
\[
p_i^2+q_i^2
=
p_i^2\left(1+\frac{1-\alpha_i^2}{\alpha_i^2}\right)
=
\frac{p_i^2}{\alpha_i^2}.
\]
Therefore, if $p_i>0$,
\[
\frac{p_i}{\sqrt{p_i^2+q_i^2}}
=
\frac{p_i}{|p_i|/\alpha_i}
=
\alpha_i,
\]
since $p_i\ge 0$. The admissibility condition $\xi_i=0$ if and only if $\alpha_i=1$ also gives $\sgn(q_i)=\xi_i$. Finally, $p_i=0$ directly implies $q_i=0$ for every admissible setting.
\end{proof}

\subsection{Proof of Lemma~\ref{lemma:k-transform}}
\label{apdx:proof-lemma-k-transform}
\begin{proof}
Fix $(\valpha,\vxi)\in\setA$ and let $\vkappa=\vk(\valpha,\vxi)$. The admissibility rule gives $\xi_i^2=1$ when $\alpha_i<1$ and $\xi_i=0$ when $\alpha_i=1$. Therefore, the inverse in \eqref{eq:apdx:k-xi-inv} recovers
\[
\frac{1}{\sqrt{1+\kappa_i^2}}=\alpha_i,
\qquad
\sgn(\kappa_i)=\xi_i.
\]
Also,
\[
\mS(\valpha,\vxi)=\mR+\mX\diag(\vkappa)=\mStilde(\vkappa).
\]
The range and robust voltage conditions that define $\setA$ consequently imply $\vkappa\in\setK$. This proves that the map is well defined and injective.

Conversely, let $\vkappa\in\setK$ and define $(\valpha,\vxi)$ by \eqref{eq:apdx:k-xi-inv}. Membership in $\Range(\vk)$ gives $\valphamin\leq\valpha\leq\vone$. It also gives $\xi_i\in\{\pm1\}$ if $\alpha_i<1$ and $\xi_i=0$ if $\alpha_i=1$. For each $i\in\setN$,
\[
\frac{\xi_i}{\alpha_i}\sqrt{1-\alpha_i^2}
=
\frac{\sgn(\kappa_i)}{(1+\kappa_i^2)^{-1/2}}
\sqrt{1-\frac{1}{1+\kappa_i^2}}
=
\kappa_i.
\]
Thus, $\vk(\valpha,\vxi)=\vkappa$ and $\mS(\valpha,\vxi)=\mStilde(\vkappa)$. The robust voltage conditions in $\setK$ then imply $(\valpha,\vxi)\in\setA$. This proves surjectivity and hence bijectivity.

Finally, the recovered identity $\alpha_i=(1+\kappa_i^2)^{-1/2}$ gives
\[
\setP(\valpha,\vxi)
=
\prod_{i=1}^n[0,\alpha_i]
=
\setP(\vkappa).
\]
For every $\vp$ in this common feasible set, $\mS(\valpha,\vxi)=\mStilde(\vkappa)$ implies
\[
w(\vp,(\valpha,\vxi))
=
\norm{\vmnom-\mStilde(\vkappa)\vp}_1
=
\wtilde(\vp,\vkappa).
\]
\end{proof}

\subsection{Proof of Lemma~\ref{lemma:compact-convex action sets}}
\label{apdx:proof-of-compact-cvx-action-sets}

\begin{proof}
The range in \eqref{eq:analy:k-range-set} is a compact convex box. For each fixed $\vp'\in\setP$,
\[
\mStilde(\vkappa)\vp'
=
\mR\vp'+\mX\diag(\vp')\vkappa
\]
is affine in $\vkappa$. Hence $\setK$ is the intersection of the range box with closed affine halfspaces indexed by $\vp'\in\setP$. It is compact and convex. It is nonempty because $\setA$ is nonempty and Lemma~\ref{lemma:k-transform} maps $\setA$ onto $\setK$.

For fixed $\vkappa$, \eqref{eq:owner-feasible-set} is a nonempty compact convex box. Let $\vd(\vkappa)$ have entries
\[
d_i(\vkappa)\triangleq\frac{1}{\sqrt{1+\kappa_i^2}}.
\]
Every $\vp\in\setP(\vkappa)$ can be written as $\vp=\diag(\vd(\vkappa))\vz$ for some $\vz\in[0,1]^n$. Therefore,
\[
I(\vkappa)
=
\max_{\vz\in[0,1]^n}
\norm{
\vmnom-\mStilde(\vkappa)\diag(\vd(\vkappa))\vz
}_1.
\]
The maximand is continuous in $(\vkappa,\vz)$, and the maximization is over a fixed compact set. The maximum theorem implies that $I$ is continuous. Since $\setK$ is nonempty and compact, $I$ attains a minimum on $\setK$.
\end{proof}

\vspace{-2em}
\begin{IEEEbiographynophoto}{Cameron Khanpour} is a Ph.D. student in the School of Electrical and Computer Engineering at the Georgia Institute of Technology, Atlanta, GA, USA. His research interests are in optimization, algorithms, and probabilistic methods for the electric grid. He received the B.S. degree in Electrical Engineering and the M.S. degree in Mathematics from the Georgia Institute of Technology. 
\end{IEEEbiographynophoto}

\vspace{-2em}
\begin{IEEEbiographynophoto}{Samuel Talkington} received the Ph.D. degree in electrical and computer engineering from the Georgia Institute of Technology, Atlanta, GA, USA, in 2026. He was a National Science Foundation Graduate Research Fellow. 

He is currently affiliated with the Department of Electrical Engineering and Computer Science at the University of Michigan, Ann Arbor, MI, USA, where he will begin an appointment as an Assistant Professor in 2027. His research interests are in efficient algorithms for decision-making and inverse problems in electric power systems.
\end{IEEEbiographynophoto}

 \vspace{-2em}
\begin{IEEEbiographynophoto}{Mathieu Dahan} received the M.S. and Ph.D. degrees in computational science and engineering from the Massachusetts Institute of Technology, Cambridge, MA, USA, in 2016 and 2019, respectively.

He is currently an Associate Professor in the School of Industrial and Systems Engineering at the Georgia Institute of Technology, Atlanta, GA, USA. His research interests are in combinatorial optimization, game theory, and predictive analytics, with applications to service and healthcare operations, humanitarian systems, and logistics and supply chain management. 
\end{IEEEbiographynophoto}

\vspace{-2em}
\begin{IEEEbiographynophoto}{Daniel K. Molzahn} is an Associate Professor in the School of Electrical and Computer Engineering at the Georgia Institute of Technology. He was a computational engineer in the Energy Systems Division at Argonne National Laboratory and a Dow Postdoctoral Fellow in Sustainability at the University of Michigan, Ann Arbor. He received the B.S., M.S., and Ph.D. degrees in electrical engineering and the Masters of Public Affairs degree from the University of Wisconsin--Madison.
\end{IEEEbiographynophoto}

\fi

\end{document}